\documentclass[letterpaper,11pt]{article}
\usepackage[margin=1in]{geometry}
\usepackage{amsmath, amsfonts, amssymb, amsthm}
\usepackage{thm-restate}
\usepackage[colorlinks=true,urlcolor=black,linkcolor=black,citecolor=black]{hyperref}
\usepackage[capitalize]{cleveref}
\usepackage{bbm,url}
\usepackage{graphicx}
\usepackage{xcolor}
\usepackage{xspace}
\usepackage{colortbl}
\usepackage{subcaption}
\usepackage{hhline}
\usepackage{pifont}
\usepackage{multirow}
\usepackage{multicol}
\usepackage{enumitem}
\usepackage{comment}
\usepackage{nicefrac,bm}
\usepackage{algorithm}
\usepackage[noend]{algpseudocode}
\usepackage{thmtools}
\usepackage{thm-restate}
\usepackage[nottoc]{tocbibind}
\usepackage{tikz}
\usetikzlibrary{positioning,arrows.meta,decorations.pathreplacing}
\usepackage[numbers]{natbib}

\allowdisplaybreaks

\newtheorem{theorem}{Theorem}
\newtheorem*{theorem*}{Theorem}

\newtheorem{lemma}{Lemma}

\newtheorem{remark}{Remark}

\newtheorem{definition}{Definition}

\newtheorem{contribution}{Contribution}

\newcommand{\s}[1]{#1}
\newcommand{\mpb}{{\textup{MPB}}\xspace}

\newcommand{\efx}{{\textup{EFX}}\xspace}
\newcommand{\pefx}{{\textup{pEFX}}\xspace}
\newcommand{\efone}{{\textup{EF1}}\xspace}

\newcommand{\pefk}{{\textup{pEF}$k$}\xspace}
\newcommand{\po}{{\textup{PO}}\xspace}
\newcommand{\fpo}{{\textup{fPO}}\xspace}

\newcommand{\x}{X}
\newcommand{\y}{Y}

\newcommand{\p}{p}

\newcommand{\N}{\mathbb N}

\newcommand{\R}{\mathbb R}

\title{Existence of 2-EFX Allocations of Chores\thanks{This work was supported by NSF Grant CCF-2334461.}}
\author{Jugal Garg\footnote{University of Illinois at Urbana-Champaign, USA} \\
\texttt{\small jugal@illinois.edu} 
\and
Aniket Murhekar\footnote{University of Illinois at Urbana-Champaign, USA}\\
\texttt{\small aniket2@illinois.edu}
}

\date{}
\begin{document}
\renewcommand{\arraystretch}{1.2}
\maketitle
\thispagestyle{empty}

\begin{abstract}

We study the fair division of indivisible chores among agents with additive disutility functions. We investigate the existence of allocations satisfying the popular fairness notion of envy-freeness up to any chore (EFX), and its multiplicative approximations. The existence of $4$-EFX allocations was recently established in~\cite{GMQ25}. We improve this guarantee by proving the existence of $2$-EFX allocations for all instances with additive disutilities. This approximation was previously known only for restricted instances such as bivalued disutilities~\cite{wu25bivalued} or three agents~\cite{afshinmehr2024approximateefxexacttefx}.

We obtain our result by providing a general framework for achieving approximate-EFX allocations. The approach  begins with a suitable initial allocation and performs a sequence of local swaps between the bundles of envious and envied agents. 
For our main result, we begin with an initial allocation that satisfies envy-freeness up to one chore (EF1) and Pareto-optimality (PO); the existence of such an allocation was recently established in a major breakthrough by~\citet{mahara2025efonepo}. We further demonstrate the strength and generality of our framework by giving simple and unified proofs of existing results, namely (i) $2$-EFX for bivalued instances \cite{wu25bivalued}, (ii) 2-EFX for three agents \cite{afshinmehr2024approximateefxexacttefx}, (iii) EFX when the number of chores is at most twice the number of agents \cite{mahara2023efxmatching}, and (iv) $4$-EFX for all instances \cite{GMQ25}. We expect this framework to have broader applications in approximate-EFX due to its simplicity and generality.
\end{abstract}


\section{Introduction}

Fair allocation is a fundamental problem studied extensively across the fields of computer science, economics, mathematics, and multi-agent systems. It is primarily concerned with finding \textit{fair} allocations of resources or responsibilities to agents with heterogeneous preferences. The allocation of indivisible items has received increasing attention in recent years, as it models several important real-world problems like course allocation, inheritance division, and task assignment. These problems can be broadly categorized depending on whether the items are \textit{goods}, i.e., provide utility or value to agents, or \textit{chores}, i.e., provide disutility or cost to agents.

In this paper, we study the fair division problem with indivisible chores where agents have \textit{additive disutility} functions. That is, the disutility an agent incurs by doing a set of chores equals the sum of disutilities of chores in that set. \textit{Envy-freeness} (EF) \cite{foleyEF} is one of the most fundamental notion of fairness. An allocation $X = (X_1, X_2, \dots, X_n)$ is said to be envy-free if every agent $i$ prefers their own bundle $X_i$ to the bundle $X_j$ assigned to any other agent $j$. Unfortunately, EF allocations are not guarantees to exist due to indivisibility ---  consider allocating a single chore among two agents. This motivated the need to adapt envy-freeness to the discrete setting. Among various  relaxations proposed, the notion closest to EF is \textit{envy-freeness up to any chore} (EFX). In an EFX allocation, every agent prefers their bundle to the bundle of another agent, after the removal of any chore from their own bundle. That is, $X$ is EFX if every agent $i$ prefers the bundle $X_i \setminus \{c\}$ to the bundle $X_j$ of any other agent $j$, for \textit{every} chore $c\in X_i$. Since EFX is a natural and compelling fairness notion in the discrete setting, investigating the existence of EFX is considered an important problem in discrete fair division \cite{procaccia20,caragiannis2019mnwacm}.

However, the existence of EFX allocations remains open for general instances with additive disutilities, even when there are only three agents. Several important works have proved the existence of EFX allocations under specific restricted domains, such as two agents or when the number of chores is at most twice the number of agents \cite{mahara2023efxmatching}; see \cref{sec:related-work} for an expanded discussion of known results. Another popular approach is to investigate the existence of \textit{approximately}-EFX allocations for general instances, without imposing any restrictions on the instance.  In a $\lambda$-EFX allocation, where $\lambda\ge 1$, the disutility of each agent $i$ after removing any chore from her own bundle is at most $\lambda$ times the disutility of $i$ for the bundle assigned to any other agent $j$.  Naturally, $\lambda=1$ corresponds to exactly EFX allocations, and $\lambda$-EFX allocations with smaller $\lambda$ are considered fairer and closer to EFX than allocations with larger $\lambda$. 

The existence of approximate-EFX allocations of chores has been investigated in several works \cite{zhou22efx,christodoulos24chores,afshinmehr2024approximateefxexacttefx,GMQ25,wu25bivalued}. The current best-known approximation is the existence of $4$-EFX allocations for all chore division instances, as established in~\cite{GMQ25}. They obtain this result by defining and proving the existence of fractional solutions called earning restricted (ER) equilibria, rounding an ER equilibrium to a desirable integral allocation, and designing involved procedures that transfer chores between agents until obtaining a $4$-EFX. Our main contribution is to significantly improve this approximation factor to $2$, through a simpler approach.

\begin{contribution}
For any chore division instance with additive disutilities, a $2$-\efx allocation always exists.
\end{contribution}

Prior to our work, the existence of $2$-\efx allocations was known only under special cases, such as bivalued disutilities \cite{wu25bivalued} or $n=3$ agents \cite{afshinmehr2024approximateefxexacttefx}. For $n=3$ agents, the authors remark ``\emph{our proof relies on an extensive case study; we raise the question of whether a more elegant proof exists.}" We answer this question positively, as our result subsumes and extends theirs to any number of agents, through a simpler approach.
Our approach is to begin with a suitable initial allocation, and perform a series of operations that involve swapping subsets of chores between envious and envied agents, until the allocation is $2$-EFX. 

The desired initial allocation satisfies a weaker notion of fairness called envy-freeness up to one chore (EF1) and the efficiency notion of Pareto-optimality (PO). An allocation is said to be EF1 if every agent $i$ prefers their bundle to that of any other agent $j$ after the removal of \textit{some} chore from their own bundle. Thus, an EFX allocation is EF1, but not vice versa. It is known that EF1 allocations can be computed in polynomial time \cite{bhaskar2020chores}. An allocation is considered PO if there is no other allocation in which some agent receives strictly lower disutility and no agent receives strictly higher disutility. The existence of allocations of chores that are simultaneously EF1 and PO was considered a challenging open problem, until its recent resolution through a remarkable breakthrough by \citet{mahara2025efonepo}. They proved that there always exists an integral allocation $X$ and a set of chore prices $\p$ such that (i) agents are only assigned chores that minimize their disutility-to-price ratio, and (ii) the allocation is \textit{price}-EF1, i.e., EF1 in terms of prices. The first condition ensures the allocation is PO (see \cref{thm:fpo}), and the second condition is stronger than EF1 (see \cref{lem:pEF1impliesEF1}). Using such an allocation as the starting point, we systematically swap subsets of bundles between agents to attain an $2$-EFX allocation.

As noted earlier, a similar approach was utilized by \cite{GMQ25} to obtain a $4$-EFX allocation by starting with an allocation that satisfies weaker guarantees than price-EF1, which was obtained by rounding a fractional ER equilibrium. We distill the core idea of the approach and find that it has much broader applicability. 
Our second contribution is to develop a general, unifying framework for obtaining approximate-EFX allocations of chores.

\begin{contribution}
We give a general framework for obtaining approximate-\efx allocations of chores to agents with additive preferences.
\end{contribution}

Our framework has two components: (1) obtaining a suitable initial allocation $X$ that satisfies certain properties, and (2) beginning with $X$, iteratively performing \textit{chore swaps} until an approximate-EFX allocation is reached. We briefly explain the two components. 

We term allocations that are suitable for finding a $\lambda$-EFX allocation as $\lambda$-EFX-\textit{friendly} allocations (formally defined in \cref{def:efx-friendly}). For a $\lambda$-EFX-friendly allocation $X$, there exists is a partition of agents into sets $N_0$ and $N_H$ such that agents in $N_0$ are $\lambda$-EFX, while agents in $N_H$ may not be. Intuitively, this implies that each envious agent $i\in N_H$ has a ``high" disutility chore $j_i$. Let $H$ be the set of such high disutility chores. We require that for all agents $i\in N$, $d_i(X_i\setminus  H) \le \lambda\cdot d_i(j)$ for all $j\in H$. That is, every agent strongly prefers their own bundle of ``low" disutility chores over \textit{any} chore in $H$. 

Next, we design an algorithm (\cref{alg:main}) that takes as input a $\lambda$-EFX-friendly allocation $X$ and iteratively performs \textit{chore swaps} between an envious agent in $N_H$ and envied agents. In each iteration, we consider an agent $i\in N_H$ who is not $\lambda$-EFX in the current allocation $Y$. Let $\ell$ be the agent $i$ envies the most. In an $(i,\ell)$ chore swap, agent $i$ picks up the entire bundle $Y_\ell$ of agent $\ell$ and transfers the chore $j_i$ to $\ell$. We prove that if we perform such swaps in a carefully chosen order, we obtain a $\lambda$-EFX allocation in at most $n$ swaps. We describe and analyze the framework in \cref{sec:framework}.

Our framework is useful as it is simple, algorithmic, and provides a unified method of obtaining approximate-EFX allocations of chores. Several works on approximate-EFX \cite{zhou22efx,christodoulos24chores,afshinmehr2024approximateefxexacttefx,GMQ25,wu25bivalued} rely on extensive case analysis and novel, instance-specific algorithmic techniques. By using our unified framework, we essentially reduce the problem of finding a $\lambda$-EFX allocation to that of finding a $\lambda$-EFX-friendly allocation. In addition to the state-of-the-art existence result for $2$-\efx which subsumes the same result for $n=3$ agents \cite{afshinmehr2024approximateefxexacttefx}, our framework provides simple, clean, and interpretable proofs of known results in a unified manner, as listed below. 
\begin{itemize}
\item Existence of $4$-EFX allocations \cite{GMQ25}. 
\item Polynomial time algorithm for computing a $(2-1/k)$-EFX allocation in a \textit{bivalued} instance, where all chore disutilities are either $1$ or $k$, for some $k\ge1$ \cite{wu25bivalued}. 
\item Polynomial time algorithm for computing an EFX allocation when the number of chores is at most twice the number of agents \cite{mahara2023efxmatching,GMQ25}.
\end{itemize}

Moreover, our framework is algorithmic. For obtaining a constant-approximate EFX allocation, we can either begin with a price-EF1 and PO allocation or by rounding the fractional earning restricted (ER) equilibrium. The proof of existence of a price-EF1 and PO allocation \cite{mahara2025efonepo} crucially uses a non-constructive fixed-point argument to find the right set of chore prices, which does not lead to efficient computation. Further, it is not known if ER equilibria can be computed in polynomial time, although they can be computed fast in practice using the Lemke's scheme for solving a linear complementarity problem (LCP). Polynomial-time algorithms for either of these two problems would immediately imply a polynomial-time algorithm for computing a constant-approximate EFX allocation. Lastly, we note that our approach is simple. In particular, we can simplify certain components of the algorithm of \cite{GMQ25} using our framework (see \cref{rem:4efxsimplification}). 

For the above reasons, we believe that our framework will find further applications in computing approximate-EFX allocations, both for general instances and in restricted domains.

\subsection{Other Related Work}\label{sec:related-work}
We mention other related work that is most relevant to \efx and \efone and \po allocations. 

\paragraph{Approximate-EFX for Chores.} We first discuss results pertaining to all instances with additive disutilities. The first non-trivial result for approximate-EFX was a polynomial time algorithm for computing an $O(n^2)$-EFX allocation for $n$ agents \cite{zhou22efx}. This was improved to $4$-EFX in \cite{GMQ25}. For superadditive disutilities, \cite{efxnonexistence2024} showed that EFX allocations do not always exist. 

Second, we discuss exact EFX allocations for restricted domains. EFX allocations exist and can be computed in polynomial time for instances with two agents, instances where agents have an identical preference order (IDO) over the chores \cite{boli2022wpropx}, or two types of chores \cite{aziz2022twotypes}. \citet{mahara2023efxmatching} gave matching-based algorithms for computing EFX allocations for instances where (i) the number of chores is at most twice the number of agents, (ii) all but one agent have IDO disutility functions \cite{mahara2023efxmatching}, and (iii) there are three agents with 2-ary disutilities \cite{mahara2023efxmatching}. For bivalued instances, EFX and PO allocations are known for $n=3$ agents \cite{GMQ23chores} or when the number of chores is at most twice the number of agents \cite{GMQ25}.

Third, we discuss approximate-EFX allocations for restricted domains. For $n=3$ agents, $2$-EFX allocations were known to exist \cite{afshinmehr2024approximateefxexacttefx}. For bivalued instances, \citet{zhou22efx} showed the existence of $(n-1)$-EFX allocations. This was improved to $3$-EFX and PO \cite{GMQ25}, which was recently improved to slightly better than $2$-EFX \cite{wu25bivalued}.

The existence of EF1 and PO allocations of chores was recently proved by \cite{mahara2025efonepo} using a novel fixed-point argument. Polynomial time algorithms are known for two agents \cite{aziz2019chores}, three agents \cite{GMQ23chores}, three types of agents \cite{GMQ24wef1po}, bivalued disutilities \cite{ebadian2021bivaluedchores,Garg_Murhekar_Qin_2022}, and two types of chores \cite{aziz2022twotypes,GMQ23chores}. \citet{GMQ25} proved the existence of 2EF2 and PO and $(n-1)$-EF1 and PO allocations for all instances. In an $\alpha$-EF$k$ allocation, for every agent $i$, the disutility of $i$ from her own bundle is at most $\alpha$ times the disutility from another agent's bundle after the removal of $k$ chores from $i$'s bundle. The above results are accompanied by polynomial time algorithms for constant number of agents, however a polynomial time algorithm for general number of agents remains open.

\paragraph{Results on Goods.} (Approximate)-EFX allocations have also been studied in the case of goods. The best-known result pertaining to all instances with additive utilities is the existence of a 0.618-EFX allocation. This approximation ratio has been improved in certain special cases \cite{barman2023parameterizedguaranteesenvyfreeallocations, amanatidis2024pushingfrontierapproximateefx}. The existence of exact EFX allocations is known for two agents, identically ordered (IDO) instances \cite{plaut2018efx}, three agents \cite{chaudhury2020efx}, two types of agents~\cite{mahara21efx}, two types of goods~\cite{gorantla23efxmultiset} and three types of agents~\cite{v2024efxexiststypesagents}. However, the existence of EFX allocations is open for $n\ge 4$ agents. Another relaxation for achieving EFX allocations is charity, where some goods are left unallocated; see e.g.,~\cite{caragiannis2019efxmnw,efxcharity2021,berger2022efx,akrami23efxsimpler}.

The existence of EF1 and PO allocations is known for goods by the seminal work of \cite{caragiannis16nsw-ef1} who showed that an allocation that maximizes the Nash welfare is both EF1 and PO. However, computing such an allocation is hard to approximate  \cite{lee2015nsw-apx,garg2017nswhardness}. \cite{Barman18FFEA} gave a pseudo-polynomial time algorithm for computing an EF1 and PO allocation. Polynomial time computation is open in full generality but is known for binary instances \cite{barman2018binarynsw} and for a constant number of agents \cite{mahara2024polynomialtimealgorithmfairefficient,GM24jair}. For bivalued instances, an EFX and PO allocation can be computed in polynomial time \cite{GM23TCS}.

We also refer the reader to detailed surveys \cite{aziz2022survey,amanatidis2023survey,liu2024mixed} for other work on discrete fair division.

\section{Preliminaries}\label{sec:prelim}

Let $[n] = \{1, 2, \dots, n\}$, for any $n\in \N$.
\paragraph{Problem Instance.} The chore division problem is to allocate a set of $M = [m]$ of $m$ indivisible chores to a set of $N = [n]$ of $n$ agents. Agent $i\in N$ has a disutility function $d_i : 2^M \rightarrow \R_{\ge 0}$, where $d_i(S)$ is the disutility (cost) incurred by agent $i$ on doing the set $S$ of chores. We assume that agents have \textit{additive disutilities}, i.e., $d_i(S) = \sum_{j\in S} d_{ij}$, where $d_{ij} > 0$ is the disutility of chore $j$ for agent $i$. An instance is said to be \textit{bivalued} if there exist $a,b\in\R_{> 0}$ such that $d_{ij} \in \{a,b\}$ for all $i\in N, j\in M$.

\paragraph{Allocation.} An allocation $\x = (\x_1, \x_2, \ldots, \x_n)$ is a partition of the chores into $n$ bundles, where $X_i$ denotes the bundle allocated to agent $i\in N$. A \textit{fractional} allocation $Y\in [0,1]^{n\times m}$ allocates chores fractionally, with $Y_{ij}$ denoting the fraction of chore $j$ allocated to agent $i$. In a fractional allocation $y$, agent $i$ receives disutility $d_i(y_i) = \sum_{j\in M} d_{ij} y_{ij}$. We assume allocations are integral unless specified.

\paragraph{Fairness and Efficiency Notions.}
An allocation $\x$ is said to be:
\begin{itemize}
\item $\lambda$-Envy-free up to $k$ chores ($\lambda$-$\textup{EF}k$) if for all $i,h\in N$, there exists $S\subseteq \x_i$ with $|S|\le k$ such that $d_i(\x_i \setminus S) \leq \lambda\cdot d_i(\x_h)$. An allocation is simply denoted by EF$k$ if it is $1$-EF$k$.
\item $\lambda$-Envy-free up to any chore ($\lambda$-\efx) if for all $i,h\in N$ and $j\in \x_i$, $d_i(\x_i \setminus \{j\}) \leq \lambda\cdot d_i(\x_h)$. An allocation is simply denoted by EFX if it is $1$-EFX.
\item Pareto optimal (\po) if there is no allocation $\y$ that dominates $\x$. An allocation $\y$ dominates allocation $\x$ if for all $i \in N$, $d_i(\y_i) \leq d_i(\x_i)$, and there exists $h \in N$ such that $d_h(\y_h) < d_h(\x_h)$. 
\item Fractionally Pareto-optimal (\fpo) if there is no fractional allocation that dominates $\x$. An \fpo allocation is clearly \po, but not vice-versa.
\end{itemize}

\paragraph{Competitive Equilibrium.} In a Fisher market for chores, each agent $i$ aims to earn an amount $e_i > 0$ by performing chores in exchange for payment. Each chore $j$ has a price\footnote{We use the term price and payment equivalently} $p_j > 0$ which specifies the payment per unit of the chore. In a fractional allocation $x$ and a price vector $\p = (p_1,\dots, p_m)$, the earning of agent $i$ is $\p(x_i) = \sum_{j \in M} p_j \cdot x_{ij}$. 

Define the \textit{minimum-pain-per-buck} (MPB) ratio of agent $i$ as $\alpha_i = \min_{j \in M} d_{ij}/p_j$. Let $\mpb_i = \{j \in M \mid d_{ij} / p_j = \alpha_i\}$ denote the set of chores which are MPB for agent $i$ for payments $\p$. We call $(x,\p)$ an MPB allocation if $x_i \subseteq \mpb_i$ for all $i\in N$.

An allocation $(x,\p)$ is a competitive equilibrium (CE) if each agent $i$ receives a bundle of lowest disutility subject to earning their requirement $e_i$, and every chore is allocated. For additive disutilities, a CE can be equivalently be characterized in terms of MPB ratios: $(x,\p)$ is a CE iff (i) for all $i\in N$, $p(x_i) = e_i$, (ii) $(x,\p)$ is an MPB allocation, and (iii) for all $j\in M$, $\sum_{i\in N} x_{ij} = 1$.

Competitive equilibria are known to have desirable efficiency and fairness properties.

\begin{restatable}[First Welfare Theorem \cite{mas1995microeconomic}]{proposition}{thmfpo}\label{thm:fpo}
Let $(x,\p)$ be an MPB allocation. Then $x$ is \fpo.
\end{restatable}

Let $\x$ be an integral allocation and $\p$ be a price vector. We let $\p_{-k}(\x_i) := \min_{S \subseteq \x_i, |S| \leq k} \p(\x_i \setminus S)$ denote the earning of agent $i$ from the bundle $\x_i$ excluding her $k$ highest paying chores. Likewise, we let $\hat{p}(\x_i) := \max_{j\in \x_i} \p(\x_i\setminus \{j\})$ denote the earning of $i$ from $\x_i$ excluding her least priced chore. The following notions are generalizations of price-EF1 (pEF1), introduced by \cite{Barman18FFEA}. 

\begin{definition}[Price \s{EF}$k$ and Price \efx]\label{def:pef1}
\normalfont An allocation $(\x, \p)$ is said to be $\alpha$-price envy-free up to $k$ chores ($\alpha$-\pefk) if for all $i, h \in N$ we have $\p_{-k}(\x_i) \leq \alpha\cdot\p(\x_h)$. Agent $i$ $\alpha$-\pefk-envies $h$ if $\p_{-k}(\x_i) > \alpha\cdot\p(\x_h)$. 

An allocation $(\x, \p)$ is said to be $\alpha$-price envy-free up any chore ($\alpha$-\pefx) if for all $i, h \in N$ we have $\p_{-X}(\x_i) \leq \alpha\cdot\p(\x_h)$. Agent $i$ $\alpha$-\pefx-envies $h$ if $\p_{-X}(\x_i) > \alpha\cdot\p(\x_h)$. 
\end{definition}

The next lemma is a sufficient condition for computing an $\alpha$-\s{EF}$k$/$\alpha$-\efx and \po allocation, and is used for computing fair and efficient allocations for both goods and chores \cite{barman2019prop1po,Barman18FFEA,ebadian2021bivaluedchores,Garg_Murhekar_Qin_2022,GMQ23chores,GMQ24wef1po,GMQ25}.
\begin{lemma}\label{lem:pEF1impliesEF1}
Let $(\x,\p)$ be an MPB allocation where $\x$ is integral. 
\begin{itemize}
\item[(i)] If $(\x,\p)$ is $\alpha$-\pefk, then $\x$ is $\alpha$-\textup{EF}$k$ and \fpo.
\item[(ii)] If $(\x,\p)$ is $\alpha$-\pefx, then $\x$ is $\alpha$-\efx and \fpo.
\end{itemize}
\end{lemma}
\begin{proof}
Since $(\x,\p)$ is an MPB allocation,  \cref{thm:fpo} shows $\x$ is \fpo. Let $\alpha_i$ be the MPB ratio of agent $i$ in $(\x,\p)$. Consider any pair of agents $i,h\in N$. 
\begin{itemize}
\item[(i)] If $(\x,\p)$ is $\alpha$-\pefk, then:
\begin{align*}
\min_{S\subseteq\x_i, |S|\le k} d_i(\x_i\setminus S) &= \alpha_i\cdot\p_{-k}(\x_i) \tag{since $(\x,\p)$ is on MPB} \\
&\le \alpha_i\cdot \alpha \cdot \p(\x_h) \tag{since $(\x,\p)$ is $\alpha$-pEF$k$ (\cref{def:pef1})} \\
&\le \alpha \cdot d_i(\x_h). \tag{since $(\x,\p)$ is on MPB}
\end{align*}
Thus, $\x$ is $\alpha$-\s{EF}$k$. 
\item[(ii)] If $(\x,\p)$ is $\alpha$-\pefx, then:
\begin{align*}
\max_{j\in\x_i} d_i(\x_i\setminus \{j\}) &= \alpha_i\cdot\p_{-X}(\x_i) \tag{since $(\x,\p)$ is on MPB} \\
&\le \alpha_i\cdot \alpha \cdot \p(\x_h) \tag{since $(\x,\p)$ is $\alpha$-\pefx (\cref{def:pef1})} \\
&\le \alpha \cdot d_i(\x_h). \tag{since $(\x,\p)$ is on MPB}
\end{align*}
Thus, $\x$ is $\alpha$-\efx. \qedhere
\end{itemize}
\end{proof}

\paragraph{Earning-Restricted Equilibrium.} 
An earning restricted (ER) competitive equilibrium \cite{GMQ25} imposes a restriction $c_j$ on the total payment chore $j$ can provide to the agents. Under these restrictions, a ER equilibrium $(X,\p)$ is a partial fractional allocation $X$ and chore prices $\p$ where every agent earns their desired earning requirement subject to  receiving a bundle of lowest disutility, and the payment from each chore is at most $c_j$. Formally, $(X,\p)$ is an ER equilibrium iff (i) for all $i\in N$, $\p(X_i) = e_i$, (ii) $(X,\p)$ is an MPB allocation, and (iii) for all $j\in M$, $\sum_{j\in M} p_j \cdot X_{ij} = \min\{p_j, c_j\}$.  

Naturally $\sum_i e_i > \sum_j c_j$, some agent will not receive their earning requirement and an ER equilibrium cannot exist. However, \cite{GMQ25} showed that $\sum_i e_i \le \sum_j c_j$ is a sufficient condition for the existence of ER equilibria.
\begin{theorem}[\cite{GMQ25}] 
An ER Equilibrium exists if and only if $\sum_{i\in N} e_i \le \sum_{j\in M} c_j$.
\end{theorem}

\section{A General Framework for Approximate-EFX}\label{sec:framework}

We present a general framework for constructing approximate-EFX allocations of chores. The framework has two key components: (1) obtaining a suitable initial allocation $X$ that satisfies certain properties, and (2) beginning with $X$, iteratively performing \textit{chore swaps} until an approximate-EFX allocation is reached.

\paragraph{Initial Allocation.} We describe the properties that the initial allocation $X$ should satisfy. 

\renewcommand{\d}{\hat{d}}

\begin{definition}[$\lambda$-EFX-friendly allocation]\label{def:efx-friendly} \normalfont
Let $X$ be an allocation with $|X_i| \ge 1$ for all $i\in N$. Let $j_i \in \arg\max_{j\in X_i} d_{ij}$, and let $S_i = X_i \setminus \{j_i\}$. Then $X$ is considered $\lambda$-\textit{EFX-friendly} if there exists a partition $N = N_0 \sqcup N_H$ of the agents that satisfy the following properties:
\begin{itemize}
\item[(i)] For all $i\in N_0$, $d_i(X_i) \le \lambda\cdot\min_{k\in N_0} d_i(X_k)$.
\item[(ii)] For all $i\in N_0$, $d_i(X_i) \le \lambda\cdot\min_{h\in N_H} d_i(j_h)$.
\item[(iii)] For all $i\in N_H$, $d_i(S_i) \le (\lambda-1)\cdot\min_{k\in N_0} d_i(X_k)$.
\item[(iv)] For all $i\in N_H$, $d_i(S_i) \le (\lambda-1)\cdot\min_{h\in N_H} d_i(j_h)$.
\end{itemize}
\end{definition}

We observe using properties (i) and (ii) that agents in $N_0$ are $\lambda$-EFX in a $\lambda$-EFX-friendly allocation. Agents in $N_H$ may not be $\lambda$-EFX due to the presence of ``high" disutility chores from the set  $H = \{j_i : i\in N_H\}$. However, for each such agent $i\in N_H$, properties (iii) and (iv) imply that the bundle $S_i$ has lower disutility (upto a $(\lambda-1)$ factor) than the single chore $j_h$ of any $h\in N_H$ and the bundle $X_k$ of any $k\in N_0$. These properties enable us to address the approximate-EFX-envy of such agents by performing chore swaps.

\paragraph{Chore Swaps.} We formally define a \textit{chore swap}, which aims to reduce approximate-EFX-envy among agents by swapping chores between the bundles of envious and envied agents. The idea of chore swaps was introduced in \cite{GMQ25}.

\begin{definition}[Chore swap]\label{def:chore-swap} \normalfont
Consider an $\lambda$-EFX-friendly allocation $X$ and an agent $i\in N_H$ who is not $\lambda$-EFX. 
Let $j_i \in \arg\max_{j\in X_i} d_{ij}$. 
Let $\ell$ be the agent who $i$ envies the most, i.e. $\ell = \arg\min\{h\in N: d_i(X_h)\}$. Then an $(i,\ell)$ swap on $X$ results in an allocation $X'$ given by $X'_i = X_i \cup X_\ell \setminus \{j_i\}$,  $X'_\ell = \{j_i\}$, and $X'_k = X_k$ for all $k\notin\{i,\ell\}$.
\end{definition}

Thus, in an $(i,\ell)$ swap, the highest disutility chore $j_i$ according to agent $i$ in $X_i$ is transferred to $\ell$ and $\ell$'s entire bundle $X_\ell$ is given to agent $i$. The key benefit of a chore swap is that it locally resolves approximate-EFX-envy of agents $i$ and $\ell$. This is clearly true for agent $\ell$ who is assigned a single chore $j_i$ in the allocation $X'$ resulting from the swap. Since $i$ is not $\lambda$-EFX in $X$, we have $d_i(X_\ell) < \lambda^{-1}\cdot d_i(X_i)$, indicating that $X_\ell$ is of sufficiently low disutility. In addition, we also know that $X_\ell$ is $i$'s least disutility bundle, and that properties (iii) and (iv) hold since $X$ is $\lambda$-EFX-friendly. We use these observations to argue that the resulting allocation $X'$ is $\lambda$-EFX for agent $i$, and moreover remains $\lambda$-EFX-friendly. In particular, we have $d_i(X'_i) \le \lambda\cdot d_i(j_i)$. We formally prove this in \cref{lem:invariants}.

A natural algorithm would then involve starting with a $\lambda$-EFX-friendly allocation $Y$, and repeatedly performing chore swaps until a $\lambda$-EFX allocation is obtained. However, it could happen that after an $(i,\ell)$ swap, agent $i$ re-develops $\lambda$-EFX-envy after a subsequent swap. Concretely, consider a subsequent swap $(h,k)$ that causes $i$ to develop $\lambda$-EFX-envy. Note that $i$ will not $\lambda$-EFX-envy $h$'s bundle after the swap, since $i$ did not $\lambda$-EFX-envy $k$'s bundle before the swap. However, $i$ could $\lambda$-EFX-envy $k$, who has the single chore $j_h$ after the swap. The above scenario indicates that performing swaps in an arbitrary order can cause $\lambda$-EFX-envy to re-develop among agents who have participated in a swap in the past. However, we show that we can prevent this from happening by re-allocating the chores $\{j_i\}_{i\in N_H}$, and performing swaps in a particular order. 

\paragraph{Re-allocating Chores and Ordering the Swaps.} We first note that $d_i(X'_i) \le \lambda\cdot d_i(j_i)$ holds in the allocation $X'$ that results from an $(i,\ell)$ swap. The main observation is that if for every subsequent $(h,k)$ swap, we have that $d_i(j_i) \le d_i(j_h)$, then $d_i(X'_i) \le \lambda\cdot d_i(j_h)$, and $i$ will not re-develop $\lambda$-EFX-envy subsequently. Motivated by this observation, we re-allocate the chores in $H = \{j_i : i\in N_H\}$ to agents in $N_H$ by using a simple round-robin procedure. We set aside the chores in $H$, arbitrarily order agents in $N_H$, and ask them to pick their least disutility chore from $H$ one by one. After such a re-allocation of $H$, we perform $(i,\ell)$ chore swaps involving $\lambda$-EFX-envious agents $i\in N_H$ picked according to the same round-robin order.

\begin{algorithm}[!t]
\caption{Computes a $\lambda$-EFX allocation from a $\lambda$-EFX-friendly allocation}\label{alg:main}
\textbf{Input:} A $\lambda$-EFX-friendly allocation $Y$\\
\textbf{Output:} A $\lambda$-EFX allocation $X$   
\begin{algorithmic}[1]

\State Let $N = N_0 \sqcup N_H$ be a partition of the agents s.t. the properties of \cref{def:efx-friendly} hold for $Y$ 
\State Let $Y_i = S_i \cup \{j'_i\}$, where $j'_i \in \arg\max_{j\in Y_i} d_{ij}$
\State Let $H \gets \{j'_i : i\in N_H\}$
\Statex \textit{--- Phase 1: Re-allocating $H$ ---}
\State Let $X_i \gets Y_i$ for all $i\in N_0$, and $X_i \gets S_i$ for all $i\in N_H$
\State Let $H'\gets H$ 
\For{$i=1$ to $|N_H|$}
\State $j_i \gets \arg\min_{j\in H'} d_{ij}$
\State $H'\gets H'\setminus \{j_i\}$
\EndFor
\Statex \textit{--- Phase 2: Chore swaps ---}
\For{$i=1$ to $|N_H|$}
\If{$i$ is not $\lambda$-EFX}
\State $\ell \gets \arg\min\{d_i(X_h): h\in N\}$
\State Perform $(i, \ell)$ swap: $X_i \gets X_i \cup X_\ell \setminus \{j_i\}$, $X_\ell \gets \{j_i\}$
\EndIf
\EndFor
\State \Return $X$
\end{algorithmic}
\end{algorithm}

\subsection{Algorithm Description and Analysis}
We put all the above ideas together and design \cref{alg:main}, which takes as input a $\lambda$-EFX-friendly allocation $Y$ and returns a $\lambda$-EFX allocation in polynomial time.

\paragraph{Algorithm Description.} Let $N = N_0 \sqcup N_H$ be a partition of the agents s.t. the properties in \cref{def:efx-friendly} hold for $Y$. Let $j'_i \in \arg\max_{j\in Y_i} d_{ij}$, and let $H = \{j'_i : i\in N_H\}$. We re-index the agents $N = [n]$ so that $N_H = [r]$ for $r = |N_H|$, and $N_0 = [n]\setminus [r]$. Phase 1 of \cref{alg:main} (Lines 4-8) re-allocates $H$ to $N_H$ using a round-robin procedure. In iteration $i\in[r]$, agent $i$ picks out the least disutility chore $j_i$ among the remaining chores in $H$. Phase 2 of \cref{alg:main} (Lines 9-12) performs chore swaps to eliminate $\lambda$-EFX-envy. In iteration $i\in [r]$, if agent $i$ is not $\lambda$-EFX, we perform an $(i,\ell)$ swap involving the agent $\ell$ who is most envied by $i$. 

\paragraph{Algorithm Analysis.}
Clearly, \cref{alg:main} terminates in at most $2n$ iterations. We prove that it terminates with a $\lambda$-EFX allocation. Let $X^i$ be the allocation just \textit{after} iteration $i$ of Phase 2 of \cref{alg:main}. Let $r = |N_H|$.
\begin{restatable}{lemma}{invariants}\label{lem:invariants}
The following hold for each iteration $i\in [r]$ of Phase 2 of \cref{alg:main}.
\begin{enumerate}[label=(\roman*)]
\item Before iteration $i$, agents in $N_H\setminus [i-1]$ do not participate in a swap. 
\item In iteration $i$, if agent $i$ participates in an  $(i,\ell)$ swap, then $i$ is $\lambda$-\efx after the swap. Moreover, $d_i(X_i^i) \le \lambda\cdot d_i(j_i)$ immediately after the swap.
\item After iteration $i$, agents in $N_0 \cup [i]$ are $\lambda$-\efx.
\end{enumerate}
\end{restatable}
\begin{proof}
We prove the above invariants inductively. Let $X^0$ be the allocation just before Phase 2 begins. We prove the above hold for $i=1$:
\begin{enumerate}[label=(\roman*)]
\item Invariant (i) holds trivially, since no agent has participated in any swap before iteration $1$.
\item Since $X^0$ is $\lambda$-EFX-friendly, properties (iii) and (iv) of \cref{def:efx-friendly} imply that $d_1(S_1) \le (\lambda-1)\cdot d_i(X^0_h)$ for all $h\in N$. Moreover, since agent $1$ is the first to pick a chore in the round-robin order, $d_1(j_1) \le d_1(j_h)$ for any $h\in [r]$. Thus,
\[
d_1(X^0_1) = d_1(S_1) + d_1(j_1) \le (\lambda-1)\cdot d_1(X^0_h) + d_1(j_h) \le \lambda\cdot d_1(X_h^0),
\]
proving that agent $1$ is $\lambda$-EFX in the allocation $X^0$. Hence, no swap takes place in iteration $1$ and invariant (ii) vacuously holds.
\item Since no swap takes place in iteration $1$, $X^1 = X^0$. Since $X^0$ is $\lambda$-EFX-friendly, properties (i) and (ii) of \cref{def:efx-friendly} imply that agents in $N_0$ are $\lambda$-EFX in $X^0$. We argued above that agent $1$ is $\lambda$-EFX in $X^0$. Thus agents in $N_0 \cup [1]$ are $\lambda$-EFX in $X^1 = X^0$, showing that invariant (iii) holds.
\end{enumerate}

Let us now assume that invariants (i)-(iii) hold for some $(i-1)\in [r-1]$, where $2\le i\le r$. We will prove that they continue to hold for $i$. 
\begin{enumerate}[label=(\roman*)]
\item Suppose an $(i-1,\ell)$ swap takes place in iteration $(i-1)$. We prove that $\ell\in N_0\cup [i-2]$.  By invariant (i), agents in $N_H\setminus [i-2]$ have not participated in a swap before iteration $i$. Hence $X^{i-2}_h = X^0_h$ for all $h\in N_H \setminus [i-2]$. 

Fix any $h\in N_H \setminus [i-2]$. We argue that agent $(i-1)$ is $\lambda$-EFX towards $h$ in $X^{i-2}$. Since $X^0$ is $\lambda$-EFX-friendly, property (iv) of \cref{def:efx-friendly} implies that $d_{i-1}(S_{i-1}) \le (\lambda-1)\cdot d_{i-1}(j_h)$. Moreover, since agent $(i-1)$ picks before agent $h$ in the round robin order, we have $d_{i-1}(j_{i-1}) \le d_{i-1}(j_h)$. Thus,
\begin{align*}
d_{i-1}(X^{i-2}_{i-1}) &= d_{i-1}(X^0_{i-1}) \\
&= d_{i-1}(S_{i-1}) + d_{i-1}(j_{i-1}) \\
&\le (\lambda-1)\cdot  d_{i-1}(j_h) + d_{i-1}(j_h) \\
&\le \lambda\cdot d_{i-1}(X_h^0) \\
&= \lambda\cdot d_{i-1}(X_h^{i-2}),
\end{align*}
implying that agent $(i-1)$ is $\lambda$-EFX towards all $h\in N_H\setminus[i-2]$ in $X^{i-2}$. Thus, $\ell \in N_0 \cup [i-2]$. Therefore, invariant (i) will hold for $i$ since no agent in $N_H\setminus [i-1]$ has participated in a swap in the first $(i-1)$ iterations.

\item Suppose an $(i,\ell)$ swap takes place in iteration $i$. We prove that agent $i$ is $\lambda$-EFX in the allocation $X^i$ after the swap.

By invariant (i), no agent in $N_H \setminus [i-1]$ has undergone any swap in the first $(i-1)$ iterations. In particular, this implies $X^{i-1}_{i} = X^0_{i} = S_{i} \cup \{j_{i}\}$. Using the fact that agent $i$ is not $\lambda$-EFX in the allocation $X^{i-1}$, we have $d_{i}(X^{i-1}_{i}) > \lambda\cdot d_{i}(X_\ell^{i-1})$. After iteration $i$, we have $X_{i}^{i} = S_{i} \cup X_\ell^{i-1}$, and $X_\ell^{i} = \{j_{i}\}$. 

Observe that:
\begin{itemize}[leftmargin=*]
\item Agent $i$ does not $\lambda$-EFX-envy bundle $X^{i}_\ell = \{j_{i}\}$. This is because: 
\begin{align*}
d_{i}(X_{i}^{i}) &= d_{i}(S_{i}) + d_{i}(X_\ell^{i-1}) \\
&< d_{i}(S_{i}) + \frac{1}{\lambda}\cdot d_{i}(X^{i-1}_{i}) \tag{using $d_{i}(X^{i-1}_{i}) > \lambda\cdot d_{i}(X_\ell^{i-1})$} \\
&= (1+\frac{1}{\lambda})\cdot d_{i}(S_{i}) + \frac{1}{\lambda}\cdot d_{i}(j_{i}) \tag{using $X^{i-1}_{i} = S_{i} \cup \{j_{i}\}$} \\
&\le (1+\frac{1}{\lambda})\cdot  (\lambda-1) \cdot d_{i}(j_{i}) + \frac{1}{\lambda}\cdot d_{i}(j_{i}) \tag{using Def.~\ref{def:efx-friendly}, property (iv)} \\
&= \lambda\cdot d_{i}(j_{i}).
\end{align*}
This also shows that after the swap, $d_{i}(X_{i}^{i}) \le \lambda\cdot d_{i}(j_{i})$, as claimed.
\item Agent $i$ does not $\lambda$-EFX-envy bundle $X^{i}_h = X^{i-1}_h$, for $h\notin\{i, \ell\}$.  By the nature of chore swaps, $X_h^{i-1}$ either contains a chore $j\in H$, or some bundle $X^0_k$ for some $k\in N_0$. Using properties (iii) and (iv) of \cref{def:efx-friendly}, this implies $d_i(S_i) \le (\lambda-1)\cdot d_i(X_h^{i-1})$. We then have:
\begin{align*}
d_{i}(X_{i}^{i}) &= d_{i}(S_{i}) + d_{i}(X_\ell^{i-1}) \\
&\le d_{i}(S_{i}) + d_i(X_h^{i-1}) \tag{by choice of $\ell$} \\
&\le (\lambda-1) \cdot d_{i}(X_h^{i-1}) + d_i(X_h^{i-1}) \tag{using Def.~\ref{def:efx-friendly}, properties (iii) and (iv)} \\
&= \lambda\cdot d_i(X_h^{i-1}) = \lambda\cdot d_i(X_h^{i}).
\end{align*}
\end{itemize}
The above two observations prove invariant (ii) holds for $i$.

\item We prove that any agent $h\in N_0 \cup [i]$ is $\lambda$-EFX in the allocation $X^i$ after the $(i,\ell)$ swap in iteration $i$. First note that as shown in invariant (ii), agent $i$ is $\lambda$-EFX in $X^i$. We therefore consider an agent $h\in N_0 \cup [i-1]$. By invariant (iii), $h \in N_0\cup [i-1]$ was $\lambda$-EFX in the allocation $X^{i-1}$ before iteration $i$. Since $X^i_k \supseteq X^{i-1}_k$ for all $k\neq \ell$, we conclude that agent $h$ does not $\lambda$-EFX-envy the bundle $X^i_k$ of any agent $k\neq \ell$ after iteration $i$ as well.

It remains to show that $h$ does not $\lambda$-EFX-envy the bundle $X^i_\ell = \{j_i\}$. We consider five cases.
\begin{itemize}[leftmargin=*]
\item If $h\in N_0$ and $h$ has not participated in a swap, then $X^i_h = X^0_h$. Property (ii) of \cref{def:efx-friendly} implies that $d_h(X^0_h) \le \lambda\cdot d_h(j_i)$. Thus, 
\[
d_h(X^i_h) = d_h(X^0_h) \le \lambda\cdot d_h(j_i) = \lambda\cdot d_h(X^i_\ell),
\]
proving that $h$ does not $\lambda$-EFX-envy $X^i_\ell$.
\item If $h\in N_0$ and $h$ has participated in swaps, then they must be swaps of the form $(k,h)$ for $k\in N_H$. After every such swap, $h$ has a single chore in $X^i$ and is EFX.
\item If $h\in N_H$ and $h$ participates in any $(k,h)$ swap in the first $i$ iterations, then $h$ has a single chore in $X^i$  and is EFX. 
\item Suppose $h\in N_H$ and $h$ does not participate in any $(k,h)$ swap. Suppose $h$ participates in a $(h,k)$ swap in iteration $h$. Since $h$ does not participate in any further swaps, we have $X_h^i = X^h_h$. By invariant (ii) we have $d_h(X^h_h) \le \lambda\cdot d_h(j_h)$. Since $h$ picks before $i$ in the round-robin order, we have $d_h(j_h) \le d_h(j_i)$. Putting the above together, we obtain
\[
d_h(X^i_h) = d_h(X^h_h) \le \lambda\cdot d_h(j_h) \le \lambda\cdot d_h(j_i) = \lambda\cdot d_h(X^i_\ell),
\]
proving that $h$ does not $\lambda$-EFX-envy $X_\ell^i$. 
\item Suppose $h\in N_H$ and $h$ does not participate in any swaps in the first $i$ iterations. Then $X^0_h = X^i_h = S_h \cup \{j_h\}$. Since $h$ picks before $i$ in the round-robin order, we have $d_h(j_h) \le d_h(j_i)$. By property (iv) of \cref{def:efx-friendly}, we have $d_h(S_h) \le (\lambda-1)\cdot d_h(j_i)$. Thus, 
\begin{align*}
d_h(X^i_h) &= d_h(X^0_h) \\
&= d_h(S_h) + d_h(j_h) \\
&\le (\lambda-1)\cdot d_h(j_i) + d_h(j_i) \\
&= \lambda\cdot d_h(j_i) \\
&= \lambda\cdot d_h(X^i_\ell),
\end{align*}
proving that $h$ does not $\lambda$-EFX-envy $X^i_\ell$.
\end{itemize}
In conclusion, $h$ is $\lambda$-EFX in the allocation $X^i$, and invariant (iii) holds. \qedhere
\end{enumerate}
\end{proof}

\noindent With \cref{lem:invariants} in hand, it is straightforward to  prove that \cref{alg:main} terminates with a $\lambda$-EFX allocation.
\begin{theorem}\label{thm:framework}
Given as input a $\lambda$-\efx-friendly allocation, \cref{alg:main} returns a $\lambda$-\efx allocation in polynomial time.
\end{theorem}
\begin{proof}
Clearly, \cref{alg:main} terminates in at most $2n$ iterations. By invariant (iii), agents in $N_0 \cup [r]$ are $\lambda$-EFX in the allocation $X^r$ at the end of iteration $r$. Since $N_H = [r]$, we conclude that $X^r$ is $\lambda$-EFX for all agents.
\end{proof}

\subsection{Obtaining Improved Guarantees}
We show that under some conditions the same framework can lead to improved EFX-approximations, if we begin with an allocation that satisfies weaker conditions than outlined in \cref{def:efx-friendly}. Let $L_i = \{j\in M : d_{ij} = \min_{k\in M} d_{ik}\}$ be the set of lowest disutility chores of an agent $i\in N$.

\begin{definition}[Weakly $\lambda$-EFX-friendly allocation]\label{def:weak-efx-friendly} \normalfont
Let $X$ be an allocation with $|X_i| \ge 1$ for all $i\in N$. Let $j_i \in \arg\max_{j\in X_i} d_{ij}$, and let $S_i = X_i \setminus \{j_i\}$. Then $X$ is considered \textit{weakly} $\lambda$-\textit{EFX-friendly} if for all $i\in N_H$, $S_i \cap L_i \neq\emptyset$, and there exists a partition $N = N_0 \sqcup N_H$ of the agents that satisfy the following properties:
\begin{itemize}
\item[(i)] For all $i\in N_0$, $\d_i(X_i) \le \lambda\cdot\min_{k\in N_0} d_i(X_k)$.
\item[(ii)] For all $i\in N_0$, $\d_i(X_i) \le \lambda\cdot\min_{h\in N_H} d_i(j_h)$.
\item[(iii)] For all $i\in N_H$, $\d_i(S_i) \le (\lambda-1)\cdot\min_{k\in N_0} d_i(X_k)$.
\item[(iv)] For all $i\in N_H$, $\d_i(S_i) \le (\lambda-1)\cdot\min_{h\in N_H} d_i(j_h)$.
\end{itemize}
\end{definition}

\noindent Similar to \cref{thm:framework}, we prove the following theorem.
\begin{theorem}\label{thm:strong}
Given as input a weakly $\lambda$-\efx-friendly allocation, \cref{alg:main} returns a $\lambda$-\efx allocation in polynomial time.
\end{theorem}
\begin{proof}
The proof of \cref{thm:strong} closely follows the proof of \cref{thm:framework} and \cref{lem:invariants}. Let $Y$ be a weakly $\lambda$-EFX-friendly allocation. Properties (i) and (ii) imply that agents in $N_0$ are $\lambda$-EFX in the allocation $Y$, leaving the algorithm to handle agents in $N_H$. For any allocation $X$ maintained during the course of the algorithm, for any agent $i$ we either have $|X_i| = 1$, or $i\in N_H$ and $S_i \subseteq X_i$. Since $i$ is EFX in the first case, we focus on the latter. Note that by definition, $S_i \cap L_i \neq \emptyset$ for any $i\in N_H$, i.e., the lowest disutility chore of $i$ in the bundle $X_i$ lies in $S_i$. This implies that $\d_i(X_i) = \d_i(S_i) + d_i(X_i\setminus S_i)$. Using this observation, by carefully revisiting the proof of \cref{lem:invariants} we can show that a weaker version of \cref{lem:invariants} holds. In particular, invariants (i), (ii), (iii) still hold, but with a weaker condition in invariant (ii) that states that $\d_i(X^i_i) \le \lambda\cdot d_i(j_i)$ holds immediately after an $(i,\ell)$ swap. This implies that \cref{alg:main} terminates with an $\lambda$-EFX allocation.
\end{proof}

As we show in Theorems~\ref{thm:bivalued} and \ref{thm:small}, improvements in the approximation factor are possible if it we begin with a weakly $\lambda$-EFX-friendly allocation $Y$ and use \cref{thm:strong}, instead of beginning with a $\lambda'$-EFX-friendly allocation and using \cref{thm:framework}, where $\lambda < \lambda'$.

\section{Applications}\label{sec:applications}

In this section, we give applications of the framework described in \cref{sec:framework}. 

\subsection{Existence of 2-EFX Allocations}\label{sec:2efx}
Our main result establishes the existence of 2-EFX allocations for \textit{all} chore division instances with additive disutilities.

\begin{theorem}\label{thm:2efx}
For any chore division instance with additive disutilities, a $2$-\efx allocation always exists.
\end{theorem}

In a recent breakthrough, \citet{mahara2025efonepo} showed the existence of EF1 and PO allocations for all chore division instances with additive disutilities. In fact, they prove the existence of an allocation $X$ and a price vector $\p$ such that $(X,\p)$ is an MPB allocation that is pEF1. We prove that $X$ is $\lambda$-EFX-friendly for $\lambda = 2$.

\begin{lemma}\label{lem:2ef}
If $(X,\p)$ is an MPB allocation that is pEF1, then $X$ is $2$-\efx-friendly.
\end{lemma}
\begin{proof}
Let us scale the disutilities to prices, i.e., scale disutilities so that $d_{ij} = p_j$ for any $j\in X_i$ and $d_{ij} \ge p_j$ for any $j\notin X_i$. Note that this is without loss of generality since the properties of EF1, EFX, fPO are scale-invariant.

Let $\rho = \min_{i\in N} \p(X_i)$ be the earning of the least earner. For each $i\in N$, let $j_i = \arg\max_{j\in X_i} p_j$ be the highest price chore in $X_i$, and let $S_i = X_i \setminus \{j_i\}$. Since $X$ is pEF1, we know $\p(S_i) = \p_{-1}(X_i)  \le \rho$, for all $i\in N$. We define $N_0 = \{i\in N : p_{j_i} \le \rho\}$, and $N_H = \{i\in N : p_{j_i} > \rho\}$. We make the following observations:
\begin{enumerate}
\item For any $i\in N_0$, we have $d_i(X_i) = \p(X_i) = \p(S_i) + p_{j_i} \le 2\rho$, since $X$ is pEF1 and $i\in N_0$.
\item For any $i\in N_H$, we have $d_i(S_i) = \p(S_i) \le \rho$, since $X$ is pEF1.
\item For any $i\in N$ and $k\in N_0$, we have $d_i(X_k) \ge \p(X_k) \ge \rho$, by using the MPB condition and the definition of $\rho$. 
\item For any $i\in N$ and $h\in N_H$, we have $d_i(j_h) \ge p_{j_h} > \rho$, by using the MPB condition and the fact that $h\in N_H$.
\end{enumerate}

It is now straightforward to show $X$ satisfies the conditions of \cref{def:efx-friendly} for $\lambda=2$.

\begin{enumerate}[label=(\roman*)]
\item Using observations (1) and (3), we have $d_i(X_i) \le 2\rho \le 2\cdot d_i(X_k)$ for any $i,k \in N_0$.
\item Using observations (1) and (4), we have $d_i(X_i) \le 2\rho \le 2\cdot d_i(j_h)$ for any $i\in N_0$ and $h\in N_H$.
\item Using observations (2) and (3), we have $d_i(S_i) \le \rho \le d_i(X_k)$ for any $i\in N_H$ and $k\in N_0$.
\item Using observations (2) and (4), we have $d_i(S_i) \le \rho \le d_i(j_h)$ for any $i, h\in N_H$. \qedhere
\end{enumerate}
\end{proof}

Since $X$ is $2$-EFX-friendly, \cref{thm:framework} implies that we can compute a $2$-EFX allocation in polynomial time from $X$. Thus, we establish the existence of $2$-EFX allocations for all instances with additive disutilities.

\subsection{Computing 4-EFX Allocations}\label{sec:4efx}

The first constant-factor approximation of EFX for allocating chores to agents with general, additive disutilities was shown by Garg, Murhekar, and Qin \cite{GMQ25}. At a high-level, their approach was to (1) define and show the existence of a fractional earning-restricted (ER) competitive equilibrium $(x^0,p)$, (2) rounding it to an integral MPB allocation $(X,p)$ that is $2$-pEF$2$, and (3) constructing an EFX re-allocation of up to $2n$ ``high" price chores, and (4) performing chore swaps on a subset of the agents. The 4-EFX guarantee is eventually obtained through intricate arguments that track the prices of bundles during the course of the algorithm. 
We re-prove their result using our framework, thereby simplifying both their algorithm and its analysis.

\begin{theorem}[\cite{GMQ25}]\label{thm:4efx}
For any chore division instance with additive disutilities, a $4$-\efx allocation always exists and can be computed in polynomial time given an earning restricted competitive equilibrium as input.
\end{theorem}

When $m\le 2n$, an EFX allocation can be computed in polynomial time \cite{mahara2023efxmatching,GMQ25}; see \cref{thm:small} for a simplified proof using our framework. Hence, we assume $m > 2n$. We then compute an ER equilibrium $(x^0,p)$ with each agent having an earning requirement $e_i=1$ and a uniform earning restriction of $\beta=\frac{1}{2}$ on each chore. Let $H=\{j\in M : p_j > \frac{1}{2}\}$ and $L = \{j\in M : p_j \le \frac{1}{2}\}$ be the set of high price and low priced chores, respectively. We scale the disutilities to prices so that $d_{ij} = p_j$ for $j\in x^0_i$ and $d_{ij} \ge p_j$ for $j\notin x^0_i$.

Garg, Murhekar, and Qin \cite{GMQ25} showed that such a solution can be rounded in polynomial time to an integral MPB allocation $(X,p)$ with the following properties:
\begin{itemize}
\item[(i)] For any $i$, $|X_i \cap H| \le 2$.
\item[(ii)] For any $i$ with $|X_i \cap H| = 2$, $p(X_i \setminus H) \le \frac{1}{2}$.
\item[(iii)] For any $i$ with $|X_i \cap H| = 1$, $p(X_i \setminus H) \le 1$.
\item[(iv)] For any $i$ with $|X_i \cap H| = 0$, $p(X_i \setminus H) \le \frac{3}{2}$.
\item[(v)] For any $i$, $\p(\x_i) \ge \frac{1}{2}$.
\end{itemize}

Due to the earning restriction on the chores, it must be that $|H| \le 2n$. Let $Z$ be an EFX allocation of the chores in $H$. Let $Y$ be the allocation given by $Y_i = (X_i \setminus H) \cup Z_i$, i.e.., the high price chores are re-allocated according to $Z$. We define $N_H = \{i\in N : |Z_i| = 1\}$ to be the set of agents who receive a single high price chore. Let $Z_i = \{j_i\}$ for each $i\in N_H$.

Define  $N_0 = N \setminus N_H$. Note that $N_0 = N_L \sqcup N_H^2$, where $N_L = \{ i\in N : Z_i = \emptyset \}$, and $N_H^2 = \{ i\in N : |Z_i| \ge 2\}$. Moreover, since $Z$ is EFX, either $N_L = \emptyset$ or $N^2_H = \emptyset$. As shown in \cite{GMQ25}, it is possible to allocate $Z$ so that $p(Y_i \setminus H) \le 1$ for every $i\in N_H^2$.

We now prove that $Y$ is $\lambda$-EFX-friendly for $\lambda = 4$ using the following observations.
\begin{enumerate}
\item Suppose $N_H^2 = \emptyset$. Then for all $i\in N_0$, we have $d_i(Y_i) = p(Y_i) \le \frac{3}{2}$ using property (iv). For $k\in N_0$, this implies $d_i(Y_i) \le 3\cdot \p(Y_k) \le 3\cdot d_i(Y_k)$ using property (iv). For $h\in N_H$, it implies $d_i(Y_i) \le 3\cdot \p(j_h) \le 3\cdot d_i(j_h)$ using the fact that $j_h \in H$. This establishes that $Y$ satisfies conditions (i) and (ii) of \cref{def:efx-friendly} for $\lambda = 3$, and hence also for $\lambda = 4$.

\item On the other hand, suppose $N_L = \emptyset$. Then for all $i\in N_0$ and $k\in N$, we have 
\begin{align*}
    d_i(Y_i) &= d_i(Y_i\setminus H) + d_i(Z_i) \\
    &\le p(Y_i \setminus H) + d_i(j_0) + d_i(Z_i\setminus j_0) \tag{here, $j_0 = \arg\min_{j\in Z_i} d_{ij}$} \\
    &\le 1 + 2\cdot \d_i(Z_i) \tag{since $|Z_i| \ge 2$} \\
    &\le 1 + 2\cdot d_i(Z_k) \tag{since $Z$ is EFX} \\
    &\le 4\cdot d_i(Z_k) \tag{since $d_i(Z_k) \ge p(Z_k) \ge \frac{1}{2}$}
\end{align*}
For $k\in N_0$, this implies $d_i(Y_i) \le 4\cdot d_i(Z_k) \le 4\cdot d_i(Y_k)$. For $h\in N_H$, it implies $d_i(Y_i) \le 4\cdot d_i(Z_k) = 4\cdot d_i(j_k)$. This establishes that $Y$ satisfies conditions (i) and (ii) of \cref{def:efx-friendly} for $\lambda = 4$.
\item For $i\in N_H$, let $S_i = Y_i \setminus H = X_i \setminus H$. We have $d_i(S_i) = p(S_i) \le \frac{3}{2}$ due to property (iii). Thus for $k\in N_0$, we have $d_i(S_i) \le \frac{3}{2} \le 3\cdot \p(Y_k) \le d_i(Y_k)$ due to property (iv). Moreover, for $h\in N_H$ we have $d_i(S_i) \le \frac{3}{2} < 3\cdot p_{j_h} \le 3\cdot d_i(j_h)$ using $j_h \in H$. This establishes that $Y$ satisfies conditions (iii) and (iv) of \cref{def:efx-friendly} for $\lambda = 4$.
\end{enumerate}
This proves that $Y$ is $4$-EFX-friendly. \cref{thm:framework} then implies that a $4$-EFX allocation can be computed in polynomial time from $Y$ using \cref{alg:main}.

\begin{remark}\label{rem:4efxsimplification}
We simplify the algorithm of \cite{GMQ25} through our approach. In \cite{GMQ25}, the order of performing chore swaps involving agents in $N_H$ is determined as follows. First, the set of chores $H' = \{j_i : i\in N_H\}$ are re-allocated to the agents in $N_H$ through a specific matching of $H'$ to $N_H$ (see Lemma 10 in the full version \cite{GMQ25arxiv}). This matching is obtained by writing an LP for min-cost matching with the weight of each edge $(i,j)\in N_H\times H'$ set as $\log d_{ij}$. The dual variables of this LP correspond to a set of chore prices $q\in R^{|H'|}$. Then, chore swaps involving agents $i\in N_H$ are carried out in non-decreasing order of $q_{j_i}$. Instead, through our approach, the re-allocation and the order of swaps is determined through a simple round-robin allocation of $H'$ to $N_H$. This also improves run-time by using a linear-time algorithm instead of solving an LP for matching.
\end{remark}

\subsection{Bivalued Instances}\label{sec:bivalued}

Recall that an instance is said to be bivalued if there exist $a,b \in  \R_{>0}$ such that $d_{ij}\in \{a,b\}$ for all $i\in N, j\in M$. We scale the disutilities so that $d_{ij} \in \{1,k\}$ for all $i,j$, where $k = \max\{b/a, a/b\}$. We refer to such an instance as a $\{1,k\}$-bivalued instance.

For $\{1,k\}$-bivalued instances, \citet{wu25bivalued} proved that a $(2-1/k)$-EFX and PO allocation can be computed in polynomial time. This improved the previous known guarantees of $3$-EFX and PO \cite{GMQ25} and $O(n)$-EFX (without PO) \cite{zhou22efx}. We re-prove this result using our framework.

\begin{theorem}[\cite{wu25bivalued}]\label{thm:bivalued}
For any chore division instance with $\{1,k\}$-bivalued disutilities, a $(2-1/k)$-\efx and \po allocation always exists, and can be computed in polynomial time.
\end{theorem}

For $\{1,k\}$-bivalued instances, it is known that an MPB allocation $(X,\p)$ that is pEF1 can be computed in polynomial time \cite{ebadian2021bivaluedchores,Garg_Murhekar_Qin_2022}. Moreover, all prices are either $1$ or $k$.  If $X$ is $(2-1/k)$-EFX, then we are done. Therefore we assume that $X$ is not $(2-1/k)$-EFX. 
We prove that $X$ is weakly $(2-1/k)$-EFX-friendly, as per \cref{def:weak-efx-friendly}.

\begin{lemma}
The above allocation $X$ is weakly $(2-1/k)$-\efx-friendly.
\end{lemma}
\begin{proof}
Let $\rho = \min_{i\in N} \p(X_i)$ be the earning of the least earner. We first prove that $\rho < k$. Since $X$ is not $(2-1/k)$-EFX, it is not $(2-1/k)$-pEFX. Therefore, there exist agents $i, h$ such that $\hat{p}(X_i) > (2-\frac{1}{k})\cdot p(X_h)$. Using the fact that $(X,p)$ is pEF1, we have $\hat{p}(X_i) \le p_{-1}(X_i) + k - 1 \le \rho + k-1$. Moreover, $p(X_h) \ge \rho$ by definition of $\rho$. Putting the above observations together, we get:
\[
\rho + k -1 \ge \ \hat{p}(X_i) > (2-\frac{1}{k})\cdot p(X_h) \ge (2-\frac{1}{k})\cdot \rho,
\]
which can only hold if $\rho < k$.

For each $i\in N$, let $j_i = \arg\max_{j\in X_i} p_j$ be the highest price chore in $X_i$, and let $S_i = X_i \setminus \{j_i\}$. Since $X$ is pEF1, we know $\p(S_i) = \p_{-1}(X_i) \le \rho$, for all $i\in N$. We define $N_0 = \{i\in N : p_{j_i} = 1\}$, and $N_H = \{i\in N : p_{j_i} = k\}$. We make the following observations:

\begin{enumerate}
\item For any $i\in N_0$, we have $\p(X_i)\le \rho+1$, since $X$ is pEF1 and $p_{j_i} = 1$. Then $\hat{p}(X_i) \le \rho$.
\item For any $i\in N_H$, we have $p(S_i) \le \rho$, since $X$ is pEF1. Then $\hat{p}(S_i) \le \max\{0,\rho-1\}$.
\item For any $i\in N_0$, we have $p(X_i) \ge \rho$ by definition of $\rho$.
\item For any $i\in N_H$, we have $p(j_i) = k > \rho$.
\end{enumerate}

We therefore have:
\begin{enumerate}[label=(\roman*)]
\item Using observations (1) and (3), we have $\hat{p}(X_i) \le \rho \le p(X_h)$ for any $i,h \in N_0$.
\item Using observations (1) and (4), we have $\hat{p}(X_i) \le \rho < k = p(j_h)$ for any $i\in N_0$ and $h\in N_H$.
\item Using observations (2) and (3), we have $\hat{p}(S_i) \le \max\{0,\rho-1\} \le \frac{\max\{0,\rho-1\}}{\rho} \cdot p(X_h)$ for any $i\in N_H$ and $h\in N_0$.
\item Using observations (2) and (4), we have $\hat{p}(S_i) \le \max\{0,\rho-1\} \le \frac{\max\{0,\rho-1\}}{k} \cdot p(j_h)$ for any $i, h\in N_H$. 
\end{enumerate}
Note that since $\rho < k$, we have $\frac{\max\{0,\rho-1\}}{k} < \frac{\max\{0,\rho-1\}}{\rho} < 1-
\frac{1}{k}$. Using the MPB condition with the above observations shows that $X$ satisfies the conditions of \cref{def:weak-efx-friendly} for $\lambda=2-1/k$.
\end{proof}

Since $X$ is weakly $(2-1/k)$-EFX-friendly, \cref{thm:strong} implies that we can compute a $(2-1/k)$-EFX allocation in polynomial time from $X$. Moreover, we can show that $X$ is PO by arguing that chore swaps maintain the MPB condition in the case of bivalued instances. Specifically, we show that in an $(i,\ell)$ chore swap in allocation $Y$, $Y_\ell \subseteq \mpb_i$ (see e.g. Lemma 22 of \cite{GMQ25arxiv} or Lemma 3.11 of \cite{wu25bivalued}. This proves \cref{thm:bivalued}. 

\subsection{Small Number of Chores}\label{sec:small-num}

When the number of chores is at most twice the number of agents, an EFX allocation can be computed in polynomial time. This was first proved by \citet{mahara2023efxmatching} using matching-based techniques, and later by \cite{GMQ25} who gave a faster and simpler algorithm using chore swaps. We re-prove this result using our framework.

\begin{theorem}[\cite{mahara2023efxmatching,GMQ25}]\label{thm:small}
For any chore division instance with additive disutilities where $m\le 2n$, an \efx allocation always exists and can be computed in polynomial time.
\end{theorem}
\begin{proof}
If $m\le n$, then any allocation $Y$ with $|Y_i| \le 1$ is EFX. Let us therefore assume $m>n$, and let $r = m-n$. Note that $r\in[n]$.

We construct a weakly $\lambda$-EFX-friendly allocation $Y$ for $\lambda=1$ using a simple, two-phase round-robin algorithm. We set $M' = M$. First, proceeding in the order $r, r-1, \dots, 1$, each agent $i$ iteratively picks their least disutility chore $e_i$ from the remaining chores $M'$. 
Then, proceeding in the order $1$ to $n$, each agent $i$ picks their least disutility chore $j_i$ in $M'$. Let $Y$ be the resulting allocation, i.e., $Y_i = \{e_i, j_i\}$ for $i\in [r]$, and $Y_k = \{j_k\}$ for $k\in N\setminus [r]$. 

We show that $Y$ is weakly $1$-EFX-friendly. Let $N_H = N$ and $N_0 = \emptyset$. Define $S_i = \{e_i\}$ for $i\in [r]$ and $S_i = \emptyset$ for $i\in N\setminus [r]$. By definition, $\d_i(S_i) = 0 \le 0\cdot d_i(j_h)$ for any $i,h\in N_H$. Thus, $Y$ is weakly $1$-EFX-friendly as per \cref{def:weak-efx-friendly}. \cref{thm:strong} then implies that an EFX allocation can be computed in polynomial time from $Y$.
\end{proof}

\section{Discussion}
In this paper, we established the existence of $2$-EFX allocations of indivisible chores to agents with additive disutilities. This was previously known only in special cases like bivalued instances \cite{wu25bivalued} or $n=3$ agents \cite{afshinmehr2024approximateefxexacttefx}, and represents a substantial improvement over the previous known existence of $4$-EFX allocations \cite{GMQ25}. We achieve this through a general, unifying framework for obtaining approximate-EFX allocations. Moreover, we obtain existing results \cite{afshinmehr2024approximateefxexacttefx,wu25bivalued,GMQ25} through our approach.

There remain many important questions for future work, such as improving the approximation factor beyond $2$. Another interesting question is if constant-EFX allocations can be computed in polynomial time. Currently, constant-EFX allocations can only be obtained by initializing our framework with a price-EF1 and PO allocation, or a suitable rounding of the earning restricted equilibrium --- both of which are not known to be polynomial time computable. However, since our framework is efficient, any advancement on the computation of such allocations would imply a polynomial time algorithm for constant-EFX. Lastly, it is interesting to see if the framework can be applied for other restricted domains, for goods, or for other concepts like charity or surplus.


\bibliographystyle{plainnat}
\bibliography{references}

\end{document}